\documentclass[11pt, a4paper, reqno]{amsart}

\usepackage{anysize}
\usepackage[utf8]{inputenc}
\usepackage{amsmath}
\usepackage{mathtools}
\usepackage{amssymb}
\usepackage{subcaption}
\usepackage{hyperref}
\usepackage{thmtools, thm-restate}
\usepackage{xcolor}
\usepackage{mdframed}
\usepackage{mathrsfs}
\usepackage{enumerate}
\usepackage{comment}
\usepackage{complexity}
\usepackage{bm}
\usepackage{quantikz}
\usepackage{tikz}
\usetikzlibrary{arrows,positioning,shapes.geometric}
\usepackage{algorithm}
\usepackage{algorithmic}

\marginsize{2.5cm}{2.5cm}{3cm}{3cm}
\newtheorem{theorem}{Theorem}

\newtheorem{corollary}[theorem]{Corollary}
\newtheorem{lemma}[theorem]{Lemma}
\newtheorem{proposition}[theorem]{Proposition}
\theoremstyle{definition}

\AfterEndEnvironment{definition}{\noindent\ignorespaces}
\AfterEndEnvironment{theorem}{\noindent\ignorespaces}
\AfterEndEnvironment{conjecture}{\noindent\ignorespaces}
\AfterEndEnvironment{lemma}{\noindent\ignorespaces}
\AfterEndEnvironment{proof}{\noindent\ignorespaces}
\AfterEndEnvironment{corollary}{\noindent\ignorespaces}
\AfterEndEnvironment{proposition}{\noindent\ignorespaces}

\newcommand{\CC}{\mathbb{C}}			% Complex numbers
\newcommand{\SU}{\mathbf{SU}}			% SU 
\newcommand{\GL}{\mathbf{GL}}			% GL
\newcommand{\norm}[1]%					% Norm
	{\left\lVert#1\right\rVert}
\newcommand{\words}[1]{\langle #1 \rangle}	% Words of a group
\newcommand{\Cl}{\mathbf{Cl}}
\renewcommand{\M}{\mathcal{M}}

\renewcommand{\K}{\mathcal{K}}

\newcommand{\tr}{\mathrm{tr}}
\newcommand{\Haar}{\mathrm{Haar}}
\newcommand{\SWAP}{\mathrm{SWAP}}

\newcommand{\defeq}{\coloneqq}			% Definition equals
\newcommand{\thmref}[1]{Theorem~\ref{#1}}	% Theorem ref
\newcommand{\corref}[1]{Corollary~\ref{#1}}	% Corollary ref

\newcommand{\appendref}[1]{Appendix~\ref{#1}}

\DeclareMathOperator{\Hom}{Hom}

\begin{document}

\title{Bounds on eventually universal quantum gate sets}

\author[Chaitanya Karamchedu]{Chaitanya Karamchedu$^\dagger$}
\address{Department of Computer Science, University of Maryland}
\email{cdkaram@umd.edu}

\author[Matthew Fox]{Matthew Fox$^\dagger$}
\address{Department of Physics, University of Colorado Boulder}
\email{matthew.fox@colorado.edu}

\thanks{$^\dagger$ These authors contributed equally.}

\author[Daniel Gottesman]{Daniel Gottesman}
\address{Department of Computer Science, University of Maryland}
\email{dgottesm@umd.edu}

\begin{abstract}
Say a collection of $n$-qu$d$it gates $\Gamma$ is \emph{eventually universal} if and only if there exists $N_0 \geq n$ such that for all $N \geq N_0$, one can approximate any $N$-qu$d$it unitary to arbitrary precision by a circuit over $\Gamma$. In this work, we improve the best known upper bound on the smallest $N_0$ with the above property. Our new bound is roughly $d^4n$, where $d$ is the local dimension (the `$d$' in qu$d$it), whereas the previous bound was roughly $d^8n$. For qubits ($d = 2$), our result implies that if an $n$-qubit gate set is eventually universal, then it will exhibit universality when acting on a $16n$ qubit system, as opposed to the previous bound of a $256n$ qubit system. In other words, if adding just $15n$ ancillary qubits to a quantum system (as opposed to the previous bound of $255 n$ ancillary qubits) does not boost a gate set to universality, then no number of ancillary qubits ever will. Our proof relies on the invariants of finite linear groups as well as a classification result for all finite groups that are unitary $2$-designs.
\end{abstract}

\maketitle

\section{Introduction}

Let $\Gamma$ be a finite subset of the special unitary group $\SU(d^n)$, where $d,n \geq 2$. We adopt a quantum computing perspective and think of $\Gamma$ as an \emph{$n$-qu$d$it gate set} so that each element of $\Gamma$ acts on an $n$-qu$d$it system whose Hilbert space is $(\CC^{d})^{\otimes n} \cong \CC^{d^n}$. We say $\Gamma$ is \emph{universal} if and only if $\Gamma$ generates a dense subset of $\SU(d^n)$ with respect to the operator norm topology. In the circuit model of quantum computation, universal gate sets play the role of the AND, OR, and NOT gates (or any other functionally complete set of Boolean logic gates) in the circuit model of classical computation. 

Interestingly, in the quantum setting a gate set $\Gamma$ need not be universal in the above sense to perform universal quantum computation (possibly in an encoded subspace \cite{JM08}). For example, $\{H, \mathrm{TOFFOLI}\}$ is not universal, but circuits over these gates can nevertheless simulate any quantum computation \cite{Shi02}. On the other hand, there exist non-universal gate sets that are classically simulable (e.g., Clifford \cite{Got98}) as well as other gate sets whose computational power is expected to lie somewhere ``in between'' the complexity classes $\BPP$ and $\BQP$ \cite{BJS11, TD04, KFG24}. Ultimately, there are many different types of non-universal gate sets, and, as stressed in \cite{AGS17}, it is both a natural and theoretically important goal to understand all the ways in which a gate set can fail to be universal. Incidentally, this goal is similar to Post's classification of all the ways in which a set of Boolean logic gates can fail to be universal \cite{Pos41}.

One reason why this goal is so challenging is because a gate set can be non-universal, despite the fact that a higher-dimensional version of it \emph{is} universal. For example, Jeandel identified a simple $6$-qubit ($d = 2$) gate set that does not densely generate $\SU(2^6)$, but which does densely generate $\SU(2^9)$ when allowed to act on a $9$-qubit system \cite{Jea04}. In fact, Jeandel's construction generalizes to $n$-qubit gate sets, and it establishes the existence of gate sets that are non-universal on fewer than $2n - 5$ qubits, but are universal on $2n - 3$ qubits. We review his construction in \appendref{appendix:jeandel}.

In this work, we are interested in this Jeandel-type of universality---hereafter called \emph{eventual universality}---in which an $n$-qu$d$it gate set is non-universal on an $n$-qu$d$it system, but is universal on an $N$-qu$d$it system for some $N \geq n$. In particular, our work builds on a paper by Ivanyos who considered the question of whether eventual universality is decidable \cite{Iva06}. Indeed, a priori, one does not know how many additional qu$d$its are needed before a given gate set might exhibit universality, so it is not clear if eventual universality is even decidable. Remarkably, however, Ivanyos proved that eventual universality is decidable. To achieve this, he bounded the number of ancillary qu$d$its one would need to add to a system before a given gate set acting on that system would exhibit universality. Specifically, he showed that an $n$-qu$d$it gate set is eventually universal if and only if it is universal on a larger, $N$ qu$d$it system, where $N \leq d^8(n-1) + 1$. 

Our main result is a significant improvement to this bound, thus improving Ivanyos' algorithm for deciding eventual universality. Our new bound is essentially a quadratic improvement and is roughly $d^4n$. For qubits, our result implies that if an $n$-qubit gate set is eventually universal, then it will exhibit universality when acting on a $16n$ qubit system, as opposed to the previous bound of a $256n$ qubit system. In other words, if adding just $15n$ ancillary qubits to a quantum system (as opposed to the previous bound of $255 n$ ancillary qubits) does not boost a gate set to universality, then no number of ancillary qubits ever will.

Our method of proof is similar to Ivanyos' and hinges significantly on the invariants of finite linear groups as well as a classification result for all finite groups that are unitary $2$-designs. However, in an effort to make this article comprehensible to the quantum computing community, we have deferred most of the technical details to the appendices.

\section{Preliminaries}

Let $\Gamma$ be an $n$-qu$d$it gate set, where $d,n \geq 2$. As mentioned in the introduction, $\Gamma$ is \emph{universal} if and only if $\Gamma$ generates a dense subset of $\SU(d^n)$ with respect to the operator norm topology. In general, $\Gamma$ is not closed under inverses, so the set it generates is merely a \emph{semigroup} in $\SU(d^n)$. However, since $\SU(d^n)$ is compact, the semigroup generated by $\Gamma$ is dense in $\SU(d^n)$ if and only if the group generated by $\Gamma$ and its inverse elements is dense in $\SU(d^n)$. For this reason, we will always assume that $\Gamma$ is closed under inverses so that $\Gamma$ generates a sub\emph{group} of $\SU(d^n)$. 

Here, we are interested in a weaker notion of universality that we call \emph{eventual universality}. Informally, this is the idea that, while $\Gamma$ itself may not be universal, a higher-dimensional variant of $\Gamma$ is. To be more precise, let $N \geq n$, let $I$ be the identity on $(\CC^{d})^{\otimes N - n}$, and let 
$$
\Gamma^N \defeq \left\{\pi (\gamma \otimes I) \pi^{-1} : \gamma \in \Gamma, \pi \in S_N\right\},
$$
where $S_N$ is the symmetric group of degree $N$. In words, $\Gamma^N$ is the set of all $N$-qu$d$it unitaries that can be made from a single element of $\Gamma$ acting on any subset of $n$ qu$d$its (in any order), with the identity acting on the remaining $N - n$ qu$d$its. As shown in \cite{Iva06}, $\Gamma^N$ is equivalently the set of all $N$-qu$d$it unitaries that can be made from a single element of $\Gamma$ and any number of $\SWAP$ gates. Given this, we say $\Gamma$ is \emph{eventually universal} if and only if there exists $N \geq n$ such that $\Gamma^N$ is universal, and we write $\K(\Gamma)$ for the smallest $N \geq n$ such that $\Gamma^N$ is universal. In case $\Gamma^N$ is not universal for all $N \geq n$, we set $\K(\Gamma) = \infty$. Thus, $\Gamma$ is eventually universal if and only if $\K(\Gamma) < \infty$.

Evidently, if $\Gamma$ is universal, then it is eventually universal. Moreover, it is known that if $\Gamma^N$ is universal, then $\Gamma^M$ is universal for all $M \geq N$ \cite{DiV95, Iva06}. However, if $\Gamma$ is eventually universal, then it is not necessarily universal. In other words, there exist $n$-qu$d$it gate sets $\Gamma$ which \textit{are} eventually universal, but for which $\K(\Gamma) > n$. Examples of such gate sets include Jeandel's construction in \cite{Jea04}, which we review in \appendref{appendix:jeandel}.

In this paper, we are interested in upper bounding $\K(\Gamma)$. Such a bound gives the maximum number of ancillary qu$d$its one would need to add to a quantum system before an eventually universal gate set $\Gamma$ exhibits universality. The first and only upper bound (as far as we know) is due to Ivanyos \cite{Iva06}, who proved that an $n$-qu$d$it gate set $\Gamma$ is eventually universal if and only if $\K(\Gamma) \leq d^8(n-1) + 1$. Here, we improve this result to roughly $d^4n$. Formally, our main result is as follows.

\begin{theorem}
\label{thm:mainthm}
    Let $\Gamma$ be an $n$-qu$d$it gate set, where $d,n \geq 2$. Then, $\Gamma$ is eventually universal if and only if $\K(\Gamma) \leq d^4(n - 1) + 1$.
\end{theorem}

The remainder of this paper is dedicated to proving this result.

\section{Main Results}

Fix $d, n \geq 2$, $N \geq n$, and let $G$ be a compact subgroup of the general linear group $\GL(d^N, \CC)$. A key notion in this work is the \emph{$2k$th moment of $G$}, 
$$
\M_{2k}(G) = \int_{g \in G} |\tr(g)|^{2k}\mu_{\Haar}(G),
$$
where $\mu_{\Haar}(G)$ is the Haar measure on $G$. Importantly, if $G$ is a compact \emph{unitary} group, then $\M_{2k}(G)$ is the \emph{frame potential} of the Haar measure on $G$ \cite{Heinrich21, Katz04, Kowalski17}. 

A priori, the various moments of $G$ are arbitrary real numbers. However, these moments actually carry a tremendous amount of information about the ``size'' of $G$. Specifically, Larsen established the remarkable fact that if the 4th moment of a compact and unitary group $G$ is a particular value, then there are few alternatives for what $G$ can be.

\begin{theorem}[Larsen's Alternative for Unitary Groups \cite{Kowalski17}]
If $G \leq \SU(d^N)$ is compact and $\M_{4}(G) = \M_4(\SU(d^N))$, then $G$ is finite or $G = \SU(d^N)$.
\end{theorem}

Larsen's alternative is useful because it implies a very simple criterion for eventual universality. To improve the readability of what follows, we slightly abuse our notation and write $\M_k(\Gamma^N)$ for $\M_k(\mathrm{cl}(\words{\Gamma^N}))$, where $\Gamma$ is a gate set, $\words{\Gamma^N}$ is the group generated by $\Gamma^N$, and $\mathrm{cl}(\words{\Gamma^N})$ is the closure of $\words{\Gamma^N}$ in $\SU(d^N)$. Note also that $\mathrm{cl}(\words{\Gamma^N})$ is compact because it is a closed subgroup of the compact group $\SU(d^N)$.

\begin{corollary}[Criterion for Eventual Universality]
\label{cor:criterion}
Let $\Gamma \subset \SU(d^n)$ be an $n$-qu$d$it gate set. Then, $\Gamma$ is eventually universal if and only if there is $N \geq n$ such that $\M_{4}(\Gamma^N) = \M_4(\SU(d^N))$ and $|\words{\Gamma^N}| = \infty$. Moreover, $\K(\Gamma) \leq N$.
\end{corollary}

In \cite{Iva06}, Ivanyos uses \corref{cor:criterion} to obtain his upper bound on $\K(\Gamma)$, and this is also our approach to improve his bound. In Ivanyos' case, however, he leverages the fact that for all compact $G \leq \SU(d^N)$, $\M_8(G) = \M_8(\SU(d^N))$ implies $\M_4(G) = \M_4(\SU(d^N))$ \cite{GT05, Heinrich21}. (In the language of unitary $t$-designs, this is simply the statement that a unitary $4$-design is a unitary $2$-design.) Therefore, it suffices to look at the $8$th moment of $G$, as opposed to the $4$th. This is a major simplification, for a result of Bannai et al. \cite{BNRT_18}, which builds on the work of Guralnick and Tiep \cite{GT05}, proves that if $d^N \geq 5$, then there are no \emph{finite} groups $G \leq \SU(d^N)$ for which $\M_8(G) = \M_8(\SU(d^N))$. Therefore, to upper-bound $N$ such that $\M_{8}(\Gamma^N) = \M_8(\SU(d^N))$ is to upper-bound $N$ such that $\M_{8}(\Gamma^N) = \M_8(\SU(d^N))$ \emph{and} $|\words{\Gamma^N}| = \infty$. In \cite{Iva06}, Ivanyos does just this and proves the following result.

\begin{theorem}[Ivanyos \cite{Iva06}]
\label{thm:ivanyoseighthmoment}
Let $\Gamma$ be an $n$-qu$d$it gate set for which there exists $N \geq n$ such that $\M_8(\Gamma^N) = \M_8(\SU(d^N))$. Then, the smallest such $N$ satisfies $N \leq d^8(n-1) + 1$. Consequently, $\K(\Gamma) \leq d^8(n-1) + 1$.
\end{theorem}

However, Larsen's alternative and \corref{cor:criterion} only call for the \emph{4th} moments to be equal, not the 8th. Thus, a better bound on the least $N$ for which $\M_4(\Gamma^N) = \M_4(\SU(d^N))$ seems plausible. Indeed, in \appendref{append:fourthmomentsboundproof}, we prove as much using similar techniques to Ivanyos.

\begin{restatable}{theorem}{fourthmoment}
\label{thm:fourthmomentbound}
Let $\Gamma$ be an $n$-qu$d$it gate set for which there exists $N \geq n$ such that $\M_4(\Gamma^N) = \M_4(\SU(d^N))$. Then, the smallest such $N$ satisfies $N \leq d^4(n-1) + 1$.
\end{restatable}

However, unlike Ivanyos' \thmref{thm:ivanyoseighthmoment}, we cannot conclude from \thmref{thm:fourthmomentbound} and \corref{cor:criterion} alone that $\K(\Gamma) \leq d^4(n-1) + 1$ because there exist \emph{finite} $G \leq \SU(d^N)$ for which $\M_4(G) = \M_4(\SU(d^N))$, e.g., the $N$-qu$d$it Clifford group $\Cl_{d}(N)$ \cite{GAE_07, Heinrich21}. For that, we need to better understand the \emph{finite} subgroups $G < \SU(d^N)$ for which $\M_4(G) = \M_4(\SU(d^N))$. 

As defined in \cite{GT05}, a \emph{finite} group $G \leq \SU(d^N)$ satisfying $\M_4(G) = \M_4(\SU(d^N))$ is called a \emph{unitary $2$-group}, which is an instance of a \emph{unitary $k$-group}. Fortunately, the properties of unitary $k$-groups are well-understood, and there is a complete classification of all unitary $2$-groups due to Bannai et al. \cite{BNRT_18, GT05}. That said, the complete classification is rather involved and includes certain irreducible representations of particular unitary and symplectic groups, as well as a finite list of exceptions. Here, we give an abridged version of this classification so to not distract from the details of the classification that matter to us. In what follows, $\overline{G}$ is the projective group $G / \mathbf{Z}(G)$, where $\mathbf{Z}(G)$ is the center of $G$, and $\Cl_d(N)$ is the $N$-qu$d$it Clifford group.

\begin{theorem}[Bannai et al. \cite{BNRT_18}, Guralnick and Tiep \cite{GT05}, Heinrich \cite{Heinrich21}, Abridged]
\label{thm:unitary-2-groups}
Let $d,N \geq 2$ such that $d^N \geq 5$ and let $G < \SU(d^N)$ be a unitary $2$-group (i.e., a finite unitary group such that $\M_4(G) = \M_4(\SU(d^N))$). Then, one of the following cases applies.
\begin{enumerate}[(i)]
\item (Lie-Type Case) $d^N$ equals $(3^k \pm 1) / 2$ or $(2^k + (-1)^k)/3$ for some positive integer $k$, and $G$ is a particular group that is not isomorphic to $\Cl_d(N)$.
\item (Extraspecial Case) $d$ is a prime power and $\overline{G}$ is isomorphic to $\overline{\Cl_d(N)}$.
\item (Exceptional Case) $d = 2$, $N = 3$, and $G$ is a particular $3$-qubit group that is not isomorphic to $\Cl_2(3)$.\footnote{That this is the only exceptional case in this abridged classification follows from the full classification in \cite{BNRT_18, GT05} together with our assumption that the dimension $d^N$ is a perfect power.}
\end{enumerate}
\end{theorem}

This classification details all the ways in which a unitary group $G$ can satisfy $\M_4(G) = \M_4(\SU(d^N))$ and $|G| < \infty$. In the context of the criterion for eventual universality (Corollary~\ref{cor:criterion}), it details all the ways in which an $n$-qu$d$it gate set $\Gamma$ can satisfy $\M_{4}(\Gamma^N) = \M_4(\SU(d^N))$ and $|\words{\Gamma^N}| < \infty$ for any given $N \geq n$. In what follows, we will use this classification to show that unless $\Gamma$ is the Clifford gate set, if $\M_{4}(\Gamma^N) = \M_4(\SU(d^N))$ and $N > 3$, then $\M_{4}(\Gamma^{N+1}) = \M_4(\SU(d^{N+1}))$ and $|\words{\Gamma^{N+1}}| = \infty$. 

First, consider the extraspecial case in \thmref{thm:unitary-2-groups}. A result by Heinrich \cite{Heinrich21} essentially ``singles out" the Clifford gate set as the unique gate set that always generates a finite group, no matter how many ancillary qu$d$its are added.

\begin{proposition}[Proposition 13.1(i) in \cite{Heinrich21}]
\label{prop:heinrich}
Let $\Gamma$ be an $n$-qu$d$it gate set, where $d,n \geq 2$. If for all $N \geq n$, $\M_4(\Gamma^N) = \M_4(\SU(d^N))$ and $|\words{\Gamma^N}| < \infty$, then $d$ is a prime power and $\overline{\words{\Gamma^N}}$ is isomorphic to the $N$-qu$d$it Clifford group $\overline{\Cl_d(N)}$. In particular, $\Gamma$ is not eventually universal.
\end{proposition}

This proposition proves that the extraspecial case in \thmref{thm:unitary-2-groups} is the unique instance for which $\Gamma$ satisfies $\M_4(\Gamma^N) = \M_4(\SU(d^N))$ for all $N \geq n$, and yet still fail to be eventually universal. Of course, it is not surprising that the Clifford group behaves this way. What is surprising, though, is that the Clifford group is the \emph{only} group that behaves this way.

On the other hand, if $\Gamma$ is such that $\words{\Gamma}$ is either Lie-type or exceptional, then the question remains \textit{how large} must $N$ be for $|\words{\Gamma^N}| = \infty$. Below, we will prove that in both cases, if $N > 4$, then $|\words{\Gamma^N}| = \infty$. We start with the Lie-type case.

\begin{proposition}
\label{prop:lietypelemma}
Let $\Gamma$ be an $n$-qu$d$it gate set, where $d,n \geq 2$. If there exists $N \geq n$ such that $\M_{4}(\Gamma^N) = \M_4(\SU(d^N))$, $|\words{\Gamma^N}| < \infty$, and $\words{\Gamma^N}$ is either Lie-type or exceptional, then $N \leq 3$ and $\K(\Gamma) \leq d^4(n-1) + 1$.
\end{proposition}

The proof idea is to exploit the dimensional requirements in the exceptional and Lie-type cases of Theorem~\ref{thm:unitary-2-groups} to obtain a restriction on $N$ and $n$ that bounds $\K(\Gamma)$. Of course, the exceptional case is ``maximally restrictive" in the sense that it only applies when $N=3$. Interestingly, the Lie-type case is similar, as the next result implies.

\begin{restatable}{lemma}{lietype}
\label{lem:lie-type-case}
Let $d, N \geq 2$. Then, there exists a positive integer $k$ such that $d^N \in \{(3^k \pm 1) / 2, (2^k + (-1)^k)/3\}$ if and only if $N = 2$ and $d \in \{2,11\}$.
\end{restatable}

We prove this in \appendref{append:lie-type-case-proof}. Using it, we can easily prove Proposition~\ref{prop:lietypelemma}.

\begin{proof}[Proof of Proposition~\ref{prop:lietypelemma}]
On one hand, it follows from Lemma~\ref{lem:lie-type-case} that $\words{\Gamma^N}$ is Lie-type only if $N = 2$. On the other hand, it follows from Theorem~\ref{thm:unitary-2-groups} that $\words{\Gamma^N}$ is exceptional only if $N = 3$. In either case, $N \leq 3$. Since $\M_4(\words{\Gamma^N}) = \M_4(\SU(d^N))$, it holds that $\M_4(\words{\Gamma^4}) = \M_4(\SU(d^4))$. Moreover, $\words{\Gamma^4}$ is neither exceptional nor Lie-type, because $4 > 3$, and $\words{\Gamma^4}$ is also not extraspecial, because $\words{\Gamma^N}$, and hence $\words{\Gamma^4}$, is not isomorphic to a subgroup of the Clifford group. These options exhaust the possibilities of $\words{\Gamma^4}$ being finite, so $|\words{\Gamma^4}| = \infty$. Consequently, $\K(\Gamma) \leq 4$. Since $2 \leq n \leq N \leq 3$ and $d \geq 2$, $d^4(n - 1) + 1 \geq 4$. Therefore, $\K(\Gamma) \leq d^4(n-1) + 1$, as desired.
\end{proof}

As a consequence of Proposition~\ref{prop:lietypelemma}, we obtain the following corollary.
\begin{corollary}
\label{cor:infiniteness}
Let $\Gamma$ be an $n$-qu$d$it gate set, where $d, n\geq 2$. If there exists $N \geq n$ such that $N \geq 4$, $\M_4(\words{\Gamma^N}) = \M_4(\SU(d^N))$, and $\words{\Gamma^N}$ is not extraspecial, then $\Gamma$ is eventually universal and $\K(\Gamma) \leq d^4(n-1) + 1$.
\end{corollary}

Altogether, these results prove \thmref{thm:mainthm}. 

\begin{proof}[Proof of \thmref{thm:mainthm}]
If $\K(\Gamma) \leq d^4(n-1) + 1$, then $\Gamma$ is eventually universal because $\K(\Gamma) < \infty$. For the other direction, suppose that $\Gamma$ is eventually universal. Then, by \corref{cor:criterion}, there exists $N \geq n$ such that $\M_4(\Gamma^N) = \M_4(\SU(d^n))$. By \thmref{thm:fourthmomentbound}, the smallest such $N$ satisfies $N \leq d^4(n-1) + 1$. For this $N$, it follows from Proposition~\ref{prop:heinrich} that $\words{\Gamma^N}$ is not extraspecial, because $\Gamma$ is eventually universal. Consequently, by Proposition~\ref{prop:lietypelemma} and Corollary~\ref{cor:infiniteness}, $\K(\Gamma) \leq d^4(n-1) + 1$, as desired.
\end{proof}

\section{Discussion}

In this work, we have improved the previously best known upper bound on the number of ancillary qu$d$its needed for an eventually universal $n$-qu$d$it gate set to exhibit universality. Our methods are similar to Ivanyos' \cite{Iva06}, who gave the first non-trivial upper bound of roughly $d^8n$. By contrast, our upper bound is essentially a quadratic improvement and is roughly $d^4n$.

Our work leaves several questions open. First, it is unclear whether our new bound is optimal. While we have, in a sense, maximally exploited Larsen's alternative in the sense that our methods use the 4th moment function (as opposed to the 8th moment function, like in \cite{Iva06}), it is conceivable that a more nuanced criterion for eventual universality could exist, and that this new criterion could support better upper bounds. 

Second, there is the related question of \emph{lower} bounds on eventual universality. These are known for some gate sets (e.g., those studied in \cite{Jea04}), however they are unknown for more general gate sets. We discuss this in more detail in \appendref{appendix:jeandel}.

Finally, the basic techniques used in this paper are applicable to non-unitary groups as well. Since post-selected quantum circuits are essentially just general linear transformations \cite{Aar04}, it could be interesting to mimic this study but for ``eventual post-selected universality''. 

We hope our work inspires more research in these directions.

\appendix

\section{Proof of \thmref{thm:fourthmomentbound}}
\label{append:fourthmomentsboundproof}

In this section, we will prove \thmref{thm:fourthmomentbound}, which we restate below for convenience. 

\fourthmoment*

Our proof of this uses techniques that are largely inspired by the methods used in \cite{Iva06}.

Recall that, conceptually, $\Gamma^N$ is the set of all $N$-qu$d$it gates formed by applying elements of $\Gamma$ to any subset of $n$ qu$d$its, and then leaving the remaining $N-n$ qu$d$its unchanged. Formally,
$$
\Gamma^N \defeq \left\{\pi (\gamma \otimes I) \pi^{-1} : \gamma \in \Gamma, \pi \in S_N\right\},
$$ 
where $S_N$ is the symmetric group of order $N$. Observe that as $N$ grows, the only aspect of $\Gamma^N$ that changes is the set of available permutations. In particular, the underlying ``fundamental" gates $\gamma \in \Gamma$ are independent of $N$. This suggests that there is a way to separate the behavior of $\Gamma^N$ as given by the elements of $\Gamma$ from the behavior of $\Gamma^N$ as given by the permutations $S_N$. Indeed, this is the essential idea underlying the following result.

\begin{lemma}(Lemma 4 in \cite{Iva06})
\label{lem:swaplem}
    Let $d,n \geq 2$ and $N \geq n$, let $\Gamma$ be an $n$-qu$d$it gate set, and let $\Sigma_N$ be a generating set of $S_N$. Then, $\langle \Gamma^N \rangle$ is dense in $\SU(d^N)$ if and only if $\langle \left( \Gamma \otimes I_{N - n} \right) \cup \Sigma_N \rangle$ is dense in $\SU(d^N)$.
\end{lemma}

Consequently, we can think of $\Gamma^N$ as simply a gate set consisting of the elements of $\Gamma$ together with a generating set of all permutations over the $N$ qu$d$its (e.g., the set of all pairwise qu$d$it SWAP gates). We adopt this interpretation of $\Gamma^N$ for the remainder of this section. 

Ultimately, this interpretation of $\Gamma^N$ will allow us to relate $\Gamma$ to a particular \textit{polynomial ideal} $J(\langle \Gamma \rangle)$ whose degree $N$ part $J_N( \langle \Gamma \rangle)$ will ``correspond'' to $\Gamma^N$. The essential idea for this comes from \textit{invariant theory}. 

To be more precise, let $m$ be an positive integer, and consider the $\CC$-vector space $\CC^{m}$ with dual space $(\CC^m)^* \defeq \Hom_{\CC}(\CC^m,\CC)$. It is an elementary fact that if $G$ is a subgroup of the general linear group $\GL(m, \CC)$, then $\CC^m$ is a left $\CC[G]$-module and $(\CC^m)^*$ is a right $\CC[G]$-module. This means that $\CC^m$ ($(\CC^m)^*$) is also an abelian group that admits left (right) scalar multiplication by elements of $\CC[G]$, the ring of polynomials with coefficients from $\CC$ and variables from $G$. 

Now suppose $G \leq \GL(m, \CC)$ acts on $\CC^m$, and let $R = \CC[x_1, \ldots, x_m]$ be the (commutative) polynomial ring on the variables $\{ x_1, \ldots, x_m \}$. Then, a polynomial $f(\mathbf{x}) \in R$ is said to be \emph{invariant} under $G$ if and only if $f(\mathbf{x}) = f(g\mathbf{x})$ for all $g \in G$. The \emph{invariant subring} of $G$, denoted $R^G$, is the subring of $R$ consisting of all the polynomials that are invariant under $G$. 

Interestingly, as the next theorem shows, for almost all positive integers $N$, the dimension of the invariant homomorphism space of $G$, i.e., $\dim_{\CC} \Hom_{\CC[G]}((\CC^m)^{\otimes N}, \CC)$, equals the size of the ``slice" of degree $N$ elements of the invariant subring $R^G$. 

\begin{theorem}[Section 3.1 in \cite{Iva06}]
\label{thm:polynomial_ideal}
Let $m$ and $n$ be positive integers, let $W = \CC^m$, let $\Gamma$ be a finite generating set of $G \leq \GL(m, \CC)$, and let $R = \CC[x_1, \ldots, x_m]$. Then, there exists a polynomial ideal $J(G) \subseteq R$, generated by homogeneous polynomials of degree $n$, such that for every $N \geq n$, 
$$
\dim_{\CC} \Hom_{\CC[\langle \Gamma^N \rangle]} \left(W^{\otimes N}, \CC \right) = \dim (R_N/J_N(G)),
$$
where $R_N$ and $J_N(G)$ denote the degree $N$ ``slice" of $R$ and $J(G)$, respectively, i.e., the homogeneous polynomials in $R$ and $J(G)$, respectively, with total degree $N$, including the zero polynomial.\footnote{Ivanyos gives this result in terms of the dimension of the quotient ring $\dim (R_N/J_N(G))$, but this is equivalent to the dimension of the invariant subring $\dim R^G_N$ due to duality, see \cite{Greub78}.}
\end{theorem}

This equivalence is the key to understanding how $\M_4( \Gamma^N)$ behaves as a function of $N$, where, recall, $\M_4(\Gamma^N)$ is our notational shorthand for $\M_4(\mathrm{cl}(\langle\Gamma^N\rangle))$. To see how this works, we start by revisiting the definition of the $2k$th moment of $G$, $\M_{2k}(G)$, which we originally defined in terms of a particular integral over the Haar measure on a compact group $G \leq \GL(m, \CC)$, 
$$
\M_{2k}(G) = \int_{g \in G} |\tr(g)|^{2k}\mu_{\Haar}(G).
$$
However, as explained in detail in \cite{Kowalski17}, and as discussed in \cite{Katz04, Qin21}, there is in fact a natural generalization of these moment functions to non-compact $G$. In particular, $\M_{2k}(G)$ has a more general interpretation as the dimension of a particular invariant space, namely, the space of $\CC[G]$-module homomorphisms from $(\CC^m \otimes (\CC^m)^*)^{\otimes k}$ to $\CC$,
$$
\M_{2k}(G) \defeq \dim_{\CC}\Hom_{\CC[G]}\left((\CC^m \otimes (\CC^m)^*)^{\otimes k}, \CC\right).
$$ 
We note that this abstract form of the moment function is precisely how Larsen's Alternative generalizes to non-compact groups. 

Combining this more general definition of the moment function with \thmref{thm:polynomial_ideal}, if $G = \langle{\Gamma^N}\rangle \leq \GL(d, \CC)$ and $W = \CC^{d} \otimes (\CC^d)^*$ so that $W^{\otimes 2} \cong \CC^{d^4}$, then for all $N \geq n$,
\begin{align*}
\M_4(\Gamma^N) &= \dim_{\CC}\Hom_{\CC[\words{\Gamma^N}]}\left(W^{\otimes 2}, \CC\right)\\
&= \dim \left(R_N/J_N(\langle \Gamma \rangle)\right),
\end{align*}
where $R = \CC[x_1, \ldots, x_{d^4}]$. Since $\M_4(\SU(d^N)) = 2$ for all $N \geq 2$ \cite{Heinrich21}, we get that $\M_4( \Gamma^N ) = \M_4(\SU(d^N))$ if and only if $\dim \left( R_N/J_N(\langle \Gamma \rangle) \right) = 2$. Therefore, to prove \thmref{thm:fourthmomentbound}, it suffices to determine the smallest $N_0$ such that for all $N \geq N_0$, $\dim \left(R_N/J_N(\langle \Gamma \rangle) \right) = 2$ (assuming, of course, that such an $N_0$ even exists).

We have now recast the proof of \thmref{thm:fourthmomentbound} into a question about the quotient of particular polynomial ideal. Therefore, we can use some tools from algebraic geometry for assistance. For a homogeneous ideal $J$, the map $N \mapsto \dim (\CC[x_1, \ldots , x_m]_N/J_N)$ is called the \textit{Hilbert function of $J$}, and it is typically denoted as $HF_J(N)$. Importantly, the Hilbert function is always ``eventually" polynomial. In other words, for all homogeneous ideals $J$, there exists a polynomial $HP_J$ (called the \emph{Hilbert polynomial of $J$}) and an integer $N_0$ such that for all $N \geq N_0$, $HF_J(N) = HP_J(N)$. The smallest $N_0$ with this property is called the \textit{index of regularity of $J$}.

In this language, then, if there exists $N \geq n$ such that $\M_4( \Gamma^N ) = \M_4(\SU(d^N)) = 2$, then the Hilbert polynomial of the ideal $J(\langle \Gamma \rangle)$ is simply the degree-$0$ polynomial $HP_{J(\langle \Gamma \rangle)}(N) = 2$.

Finally, in the particular case that the Hilbert polynomial of an ideal $J \subseteq \CC[x_1, \ldots, x_m]$ is \textit{constant}, Lazard proved that if $J$ is generated by homogeneous polynomials of degree $n$, then the index of regularity is bounded above by $m(n-1)+ 1$ \cite{Lazard81, Lazard01}. In our case, $J(\langle \Gamma \rangle)$ is a polynomial ideal in $R = \CC[x_1, \ldots, x_{d^4}]$, and is indeed generated by homogeneous polynomials of degree $n$. Therefore, Lazard's bound shows that if there exists an $N$ such that $\M_4( \Gamma^N ) = \M_4(\SU(d^N)) = 2$, then $N \leq d^4(n-1) + 1$. This completes the proof of \thmref{thm:fourthmomentbound}.

\section{Proof of Lemma~\ref{lem:lie-type-case}}
\label{append:lie-type-case-proof}

In this section, we will prove Lemma~\ref{lem:lie-type-case}, which we restate below for convenience. 

\lietype*

Our proof relies on three lemmas to do with the theory of Diophantine equations.

\begin{lemma}[Theorem 3 in \cite{BM02}]
\label{BM02_thm3}
The equation 
$$
y^q = \frac{x^n - 1}{x - 1}
$$
has only three solutions in integers with $2 \leq x \leq 10^6$, $y > 1$, $n > 2$, and $q \geq 2$, namely, $(x,y,n,q)$ is either $(3,11,5,2)$, $(7,20,4,2)$, or $(18,7,3,3)$.
\end{lemma}

\begin{lemma}[Theorem 2 in \cite{BM02}]
\label{BM02_thm2}
The equation 
$$
y^q = \frac{x^n + 1}{x + 1}
$$
has no solution in integers with $2 \leq x \leq 10^4$, $n \geq 5$ odd, $y >1$, and $q \geq 2$.
\end{lemma}

\begin{lemma}[Lemma in \cite{Cohn94}]
\label{Cohn}
The equation $y^2 - 2z^k = - 1$ has only two solutions in integers with $k > 2$, namely, $(y,z,k)$ is either $(239, 13, 4)$ or $(1,1,k)$.
\end{lemma}

We now prove Lemma~\ref{lem:lie-type-case}.

\begin{proof}[Proof of Lemma~\ref{lem:lie-type-case}]
We will show that $N = 2$ and $d \in \{2,11\}$ are the only possibilities via a case-by-case study.

\

\textbf{Case 1:} Suppose $d^N = (3^k - 1)/2$. Then,
$$
d^N = \frac{3^k - 1}{3-1}.
$$
By Lemma~\ref{BM02_thm3}, the only solution in integers to this equation with $k \geq 2$ is $(d,N,k) = (11,2,5)$. If $k = 2$, then the only solution is $(d,N,k) = (2,2,2)$. By Lemma~\ref{BM02_thm3}, these are the only solutions.

\

\textbf{Case 2:} Suppose $d^N = (3^k + 1)/2$. Then, 
$$
2d^N = 3^k + 1.
$$
We will show that there are no integer solutions in $d, N,$ and $k$ with $d,N \geq 2$. If $d$ is even, then $2d^N = 0 \pmod{8}$, however $3^k + 1 \in \{2,4\} \pmod{8}$. Therefore, $d$ is odd. Since $3^k + 1 = 1 \pmod{3}$, $2d^N = 1 \pmod{3}$ as well, so $d^N \equiv 2 \pmod{3}$. Consequently, $d \not\in \{0, 1\} \pmod {3}$, which is to say that $d = 2 \pmod{3}$. Therefore, $d^N = 2^N = 2 \pmod{3}$, which implies that $N$ is odd. Since $d$ is also odd, $d^N \in \{1,3,5,7\} \pmod{8}$, so $2d^N \in \{2,6\} \pmod{8}$. However, $3^k + 1 = 2 \pmod{8}$ if $k$ is even, and $3^k + 1 = 4 \pmod{8}$ if $k$ is odd. Thus, $k$ is even, so $k = 2\ell$ for some integer $\ell$. Rearranging the equation in Case 2, we get that $2d^N - 3^{2\ell} = 1$, or equivalently, 
$$
2d^N - (3^\ell)^2 = 1.
$$
By Lemma~\ref{Cohn}, there are no integer solutions to this expression with $d \geq 2$ and  $N \geq 3$. Finally since we know from above that $N$ must be odd, $N = 2$ can also not yield a valid solution, so there are no integer solutions to the equation $d^N = (3^k + 1)/2$ with $d \geq 2$ and $N \geq 2$.

\

\textbf{Case 3:} Suppose $d^N = (2^k - (-1)^k)/3$. We will show that there are no integer solutions in $d, N,$ and $k$ with $d,N \geq 2$. On one hand, if $k$ is even, then $k = 2\ell$ for some integer $\ell$. Thus,
$$
d^N = \frac{2^k - 1}{3} = \frac{2^{2\ell} - 1}{3} = \frac{4^\ell - 1}{4-1}.
$$
By Lemma~\ref{BM02_thm3}, there are no integer solutions to this equation with $\ell \geq 3$. It is straightforward to check that $\ell \in \{1,2\}$ do not yield valid solutions either. On the other hand, if $k$ is odd, then 
$$
d^N = \frac{2^k + 1}{3} = \frac{2^k + 1}{2 + 1}.
$$
By Lemma~\ref{BM02_thm2}, there are no integer solutions to this equation with $k \geq 5$. Since $d,N \geq 2$, it is straightforward to check that $k = 3$ does not yield a valid solution either.

Altogether, the only valid solution to the premise of Lemma~\ref{lem:lie-type-case} derives from Case 1, where $N = 2$ and $d \in \{2,11\}$. This is the desired result.
\end{proof}

\section{Jeandel's Construction}
\label{appendix:jeandel}

Here, we review the main idea in Jeandel's paper \cite{Jea04}, which not only establishes the existence of eventually universal $n$-qu$d$it gate sets $\Gamma$ with $\K(\Gamma) > n$, but which also gives a general method to construct $n$-qubit ($d = 2$) gate sets $\Gamma$ for which $2n - 5 \leq \K(\Gamma) \leq 2n - 3$.

Let $\Omega$ be a universal $2$-qubit gate set with elements $A_1, A_2, \dots, A_{|\Omega|}$, and suppose that for all $i$, $A_i^2 = I$, where $I$ is the identity operation. For any positive integer $k \geq 2$, we define a $(k+2)$-qubit gate set $\Gamma$ implicitly as follows: $B_{k,i} \in \Gamma$ if and only if for all $\ket{t} \in \CC^{4}$ and all $\ket{c} \in \CC^{2^k}$,
$$
B_{k,i} (\ket{t} \otimes \ket{c})
=
\begin{cases} 
(A_i \ket{t}) \otimes \ket{c} & \text{if }\ket{c} \in \big\{\ket{0}^{\otimes k}, \ket{1}^{\otimes k}\big\}, \\
\ket{t} \otimes \ket{c} & \text{otherwise}. 
\end{cases}
$$
Conceptually, $B_{k,i} \in \Gamma$ provided it applies $A_i \in \Omega$ to the first two qubits if and only if the latter $k$ qubits are either all $\ket{0}$ or all $\ket{1}$.

We claim that $\Gamma$ is not universal on fewer than $2k - 2$ qubits. To see this, consider the action of $\Gamma$ on the subspace spanned by $\ket{0}^{k-1} \otimes \ket{1}^{k-1}$ up to permutations of the qubits (i.e., any computational basis state with $k-1$ $\ket{0}$'s and $k-1$ $\ket{1}$'s). By construction, no subset of $k$ qubits satisfies the control conditions of the individual gates $B_{k,i}$, so every such gate leaves this subspace invariant. As such, $\Gamma$ is not universal on $2k - 2$ qubits.

On the other hand, at least for some specific values of $k$, $B_{k,i}$ \emph{is} universal on $2k+1$ qubits. To see this, consider a set of $2k+1$ qubits, and suppose that we want to apply $A_i$ to the first two qubits. To do this, we need to act $B_{k,i}$ on a $k+2$ qubit subsystem that includes the first two qubits, as well as $k$ control qubits. These control qubits are selected as a subset of the remaining $2k-1$ qubits. However, we do not know the state of those $2k-1$ qubits, and so a priori we do not know which subset of $k$ qubits to select as the controls. So, instead of selecting any particular subset, we will simply try every subset, and apply the $B_{k,i}$ gate $\binom{2k-1}{k}$ times. The question, then, is how many times is the gate $A_i$ applied to the first 2 qubits? We will show that for an appropriate choice of $k$, no matter the state over the $2k-1$ qubits, $A_i$ will be applied exactly once on the first two qubits.

Let $\ket{\psi}$ be any computational basis state on $2k-1$ qubits. Then $\ket{\psi}$ contains at least $k$ tensor factors of either $\ket{0}$ or $\ket{1}$. Without loss of generality, suppose that there are that there are $k + q$ tensor factors of $\ket{0}$, where $0 \leq q \leq k-1$. Then, the number of times that the gate $A_i$ is applied on the first two qubits is $\binom{k+q}{k}$. The key observation is that if $k = 2^j$ and $q \leq k-1$, then $\binom{k+q}{k}$ is odd. (This follows from inducting on $j$.)

Thus, by the reasoning above, for any computational basis state $\ket{\psi}$ on $2k-1$ qubits, applying $B_{k,i}$ $\binom{2k-1}{k}$ times, once for every subset of the $k$ control qubits, results in $A_i$ being applied exactly once to the first two qubits, provided $k$ is a power of two. Since this will be true for any computational basis state over the $2k-1$ qubits, this construction will hold for any superposition state over them as well. 

If we now want to apply the gate $A_i$ on a different set of two qubits, we just separate those two qubits as ``the first two" and repeat the exact same process as described above. As such, the gates in $\Omega$ can be applied on \emph{any} of the $2k+1$ qubits, which proves that $\Gamma$ is universal on $2k + 1$ qubits.

Altogether, then, we have shown that there exists a $(k+2)$-qubit gate set $\Gamma$ for which $2k - 1 \leq \K(\Gamma) \leq 2k + 1 $. With $n = k + 2$, we have equivalently shown the existence of an $n$-qubit gate set $\Gamma$ for which $2n - 5 \leq \K(\Gamma) \leq 2n - 3$.

We note that by basically the same argument above, one can establish the existence of an $n$-qu$d$it gate $\Gamma$ for which $\K(\Gamma) \geq dn - 2d - 1$. Therefore, there exist $n$-qu$d$it gate sets that are not universal on $n$-qu$d$its. However, it remains to show that such gates sets are also eventually universal. Unfortunately, the upper bound argument above does not obviously generalize to qu$d$it systems. This is because the proof for the upper bound uses the fact that when $d = 2$, one can count the number of times a gate $B_{k,i}$ is ``activated" via a single binomial coefficient. When $d > 2$, however, this count is a complicated sum of binomial coefficients whose parity is not easily deducible. Thus, the argument does not go through, at least not obviously. Still, we conjecture that for all $d > 2$, there exists an \emph{eventually universal} $n$-qu$d$it gate set $\Gamma$ for which $dn - 2d - 1 \leq \K(\Gamma)$, or something morally equivalent. We leave this as an open question.

\section*{Acknowledgments}

The authors thank Gabor Ivanyos, Jonathan Rosenberg, Luke Schaeffer, Amin Gholampour, Ian Teixeira, and Adam Bouland for several useful discussions.

%% Begin Bibliography %%
\bibliographystyle{amsplain}
\bibliography{references}

\end{document}